\documentclass[jmp,aps, amsmath,amsthm, amssymb,reprint, showkeys, showpacs]{revtex4-1}

\usepackage{graphicx}
\usepackage{dcolumn}
\usepackage{bm}

\usepackage{graphicx}
\usepackage{epsfig}

\newtheorem{Thm}{Theorem}[section]
\newtheorem{theorem}[Thm]{Theorem}

\newtheorem{lemma}[Thm]{Lemma}

\newtheorem{remark}{Remark}[section]

\newenvironment{proof}[1][Proof]{\begin{trivlist}
\item[\hskip \labelsep {\bfseries #1}]}{\end{trivlist}}

\newcommand{\qed}{\nobreak \ifvmode \relax \else
      \ifdim\lastskip<1.5em \hskip-\lastskip
      \hskip1.5em plus0em minus0.5em \fi \nobreak
      \vrule height0.75em width0.5em depth0.25em\fi}

\def\unit{\bbbone}
\def\bbbone{{\mathchoice {\rm 1\mskip-4mu l} {\rm 1\mskip-4mu l}
{\rm 1\mskip-4.5mu l} {\rm 1\mskip-5mu l}}}

\begin{document}

\title{On the mixing property for a class of states
of relativistic quantum fields}


\author{Christian D.\ J\"akel}
\email{christian.jaekel@mac.com}
\affiliation{School of  Mathematics, Cardiff University, Wales, \\
CF24 4AG, United Kingdom.}

\author{Heide Narnhofer}%
 \email{heide.narnhofer@univie.ac.at}
\affiliation{ Mathematical Physics,  
Universit\"at Wien, \\
Boltzmanngasse 5,
1090 Vienna, Austria.
}%

\author{Walter F. Wreszinski}
 \email{wreszins@gmail.com, supported in part by CNPq}
\affiliation{%
Departamento de F\'isica Matem\'atica,  
Instituto de F\'isica, \\
USP, Caixa Postal 66318  
05314-970, S\~ao Paulo, Brazil.}

\begin{abstract}
Let $\omega$ be a factor state on the quasi-local algebra $\cal{A}$ of observables generated by a relativistic quantum field, which in addition satisfies
certain  regularity conditions (satisfied by ground states and the recently constructed thermal states of the $P(\phi)_2$ theory). We prove that there exist space and time translation invariant states, some of which are arbitrarily close to $\omega$ in the weak* topology, for which the time evolution is weakly asymptotically abelian. 
\end{abstract}

\maketitle

\section{Introduction and Summary}
\label{Sec1}

Let $({\cal A},\tau)$ be a $C^*$- or $W^*$-dynamical system (see, e.g.~\cite{1,2}), where ${\cal A}$ is a quasi-local algebra \cite[Vol.~1, Sec.~2.6]{1} and~$\tau$~is a group of time-translation automorphisms of ${\cal A}$. Let $\omega$ be
a $\tau$-invariant state, assumed to be normal in the $W^*$-case.
The GNS  triple $({\cal H}_{\omega},\pi_{\omega}, \Omega_{\omega})$ associated to the pair
$({\cal A}, \omega)$ consists \cite[Theorem 2.3.16, Vol.~1]{1} of a (separable) Hilbert space~${\cal H}_{\omega}$,
a representation~$\pi_{\omega}$ of ${\cal A}$ on  ${\cal H}_{\omega}$, and a
vector~$\Omega_{\omega}$, which is cyclic for~$\pi_{\omega}({\cal A})$.
The representation $\pi_{\omega}$
maps the triple~$({\cal A},\tau,\omega)$ into a new triple
$({\cal R}_{\omega},\tilde{\tau},\tilde{\omega})$, a $W^*$-dynamical system on the enveloping von Neumann algebra
${\cal R}_{\omega}=\pi_{\omega}({\cal A})''$, with a normal invariant state
\begin{equation}
\tilde{\omega}(A)=(\Omega_\omega, A\Omega_\omega),
\qquad
A \in  {\cal R}_{\omega}\, .
\label{1}
\end{equation}
$(\, . \, ,\, . \, )$ denotes the scalar product in ${\cal H}_{\omega}$.
Since $\omega$ is $\tau$-invariant, the $W^*$-dynamics~$\tilde{\tau}$,
\begin{equation}
\tilde{\tau}_t(A)=U_{\omega}^{t}A(U_{\omega}^{t})^{*} ,
\label{2}
\end{equation}
is implemented by a one-parameter group $\{ U_{\omega}^{t} \mid t \in \mathbb{R} \}$
of unitary operators
\begin{equation}
U_{\omega}^{t}=\exp(itL_{\omega})
\label{3}
\end{equation}
acting on ${\cal H}_{\omega}$.
The self-adjoint operator $L_{\omega}$, known as the $\omega$-Liouvillean \cite{2}, has the property
\begin{equation}
L_{\omega}\Omega_{\omega}=0 \, .
\label{4}
\end{equation}
Quantum versions of the ergodic theorems of classical dynamics \cite{3} have been formulated and proved \cite{2, 4}. The natural quantum counterpart of the definition of ergodicity in classical mechanics is the notion of pure state \cite[Vol.~1, pg.~52, Def.~2.3.14]{1} or primary or factor state \cite[Vol.~1, pg.~81]{1}. The Koopman-von Neumann spectral characterisation of ergodicity \cite{3} has the following quantum analogue: $\omega$ is ergodic iff  $\hbox{Ker } L_{\omega}$ is one-dimensional, or, equivalently,
\begin{eqnarray}
&\lim_{T\to\infty}\frac{1}{T}\int_0^T\, {\rm d}t \; [\omega(A\tau_t(B))  -\omega(A)\omega(B)]=0
\qquad
\label{5}
\end{eqnarray}
for all $A,B\in{\cal A} $.
Again in close analogy to classical dynamics \cite{3}, a $C^*$-dynamical system $({\cal A},\tau,\omega)$ is
said to be mixing iff (see, e.g.,~\cite{2}):
\begin{equation}
\lim_{t\to\infty}\omega(A\tau_t(B))-\omega(A)\omega(B)=0
\qquad
\forall A,B\in{\cal A} \, .
\label{6}
\end{equation}

In classical mechanics, it is well-known that the mixing property (\ref{6}) has a much more dramatic effect than ergodicity~(\ref{5}): it represents the first step in a ergodic hierarchy crowned by Bernoulli or K-systems, which display fully chaotic behaviour \cite{3}. The quantum theory of the latter has been developed in Ref.~\cite{5}.
A necessary and sufficient condition for mixing is~\cite{2}:
\begin{equation}
w-\lim_{t\to\infty}\exp(itL_{\omega})=\Omega_{\omega}(\Omega_{\omega},\, . \, ) \, .
\label{7}
\end{equation}
A sufficient condition for mixing, which follows from the Riemann-Lebesgue lemma, is (see, again, \cite{2}):

\begin{lemma} \label{8}
If the spectrum of $L_{\omega} $ on $ \Omega_{\omega}^{\bot}$ 
is purely absolutely continuous, then $({\cal A},\tau,\omega) $ is mixing. 
\end{lemma}

The fact that the condition stated in Lemma \ref{8} is not necessary is due to the existence of singular continuous measures whose Fourier transform decays at infinity---the so-called Rajchman measures \cite{6}.
In spite of the great conceptual and practical importance of the mixing condition (\ref{6}), it has been seldom studied in quantum field theory. One exception is \cite{7}, the other is \cite{8} (see also \cite{9, 10, 11}). Here one must distinguish the vacuum state, for which we need only consider a $C^*$-dynamical system and~$L_{\omega}$ should be identified with the physical Hamiltonian~$H_{\omega}$, and thermal states, which satisfy the KMS condition.
%

In a beautiful paper, Maison \cite{7} proved that, in a unitary representation of the
Poincar\'e group~$P^\uparrow_+$, the infinitesimal space-time translations have a spectral measure without singular continuous part. By Lemma \ref{8}, this implies (\ref{6}) for the ground state, under the assumption that the latter is invariant under the group of Poincar\'e automorphisms: this yields a unitary representation of the Poincar\'e group by a well-known argument, already used to establish~(\ref{2}) (see, e.g.~\cite[Vol.~1, Corollary 2.3.17, pg.~56]{1}).

In recent years there arose a special interest in thermal quantum field theory \cite{12}, which is expected to play an important role in cosmology (see the concluding remarks of Section~\ref{Sec4}). Thermal states of quantum fields are not, however, invariant under Lorentz boosts, because the KMS condition distinguishes a rest frame (see also the discussion in \cite{12,13}).

In this paper, we generalise the theorem of \cite{7} to a class which includes thermal quantum fields. The method is entirely different from Maison's, which is based on the structure of the irreducible representations of the Poincar\'e group: it consists of an extension of the arguments introduced in \cite{8,10,11}, and applied there to certain Galilei invariant theories. The basic idea of \cite{8} was to exploit that
\begin{itemize}
\item[(a)] the boost relates space-translations and time-translations;
\item[(b)] for space-translations, the large-distance behaviour is under control for the class of models considered.
\end{itemize}
Two new elements of the present extension are:
\begin{itemize}
\item[(c)] the introduction of a time-dependent scale in the boosts;
\item[(d)] the explicit use of local commutativity.
\end{itemize}
In Section~\ref{Sec2} we present our framework, consisting of Assumptions {\bf A1-A5}. The known examples included in this framework are also briefly reviewed there. In Section~\ref{Sec3}, we prove~(\ref{6}) for a dense set (in the weak* topology) of time- and space-translation invariant states of a (relativistic) quantum field theory satisfying the assumptions of Section~\ref{Sec2} (Theorem~1, Theorem~2 and Theorem~3).
%
%
Section~\ref{Sec4} is reserved to the conclusion, open problems and conjectures.

\section{The framework: assumptions and examples}
\label{Sec2}

We denote by $(x^{\mu})$, $\mu=0,1,\ldots,\nu$, the points of Minkowski space-time
${\mathbb R}^{1+\nu}$. Thus $\nu$ is the space dimension and $x^{0}=t$ ($c=1$) denotes the time-variable. The transformations~$T(a)$ and~$L(v)$ of ${\mathbb R}^{1+\nu}$, corresponding to space-time translation by $a \in {\mathbb R}^{1+\nu}$ and velocity boost by $u \in(-1,1)$ 
along \cite{37} the $x^1$-axis, are defined, respectively, by
\begin{equation}
T(a)x=x+a
\label{9a}
\end{equation}
and
\begin{equation}
L(v)x=
\begin{pmatrix}
\, x^{0} \cosh v- x^{1}\sinh v\, \\
x^{1}\cosh v- x^{0}\sinh v\\
 x^{2}\\
 \vdots\\
 x^{\nu} 
\end{pmatrix}
\label{9b}
\end{equation}
where $u = \tanh v$ and
$\cosh v =(1-u^{2})^{-1/2}$. The corresponding automorphisms of ${\cal A}$, denoted  by
$\xi(a)\equiv(\tau_{a^0},\sigma_{\vec{a}})$  and $\lambda_v$,  satisfy the relations
\begin{equation}
\xi_{ (t,\vec{x}) } =\tau_{t} \circ \sigma_{\vec{x}}=\sigma_{\vec{x}} \circ \tau_{t}
\qquad
\forall t\in {\mathbb R},
\quad
\forall \vec{x}\in {\mathbb R}^{\nu},
\label{10}
\end{equation}
and
\begin{equation}
\lambda_{v} \circ \xi_{a} \circ \lambda_{-v}=\xi_{ L(v)a }
\qquad
\forall a\in {\mathbb R}^{1+\nu},
\quad
\forall v\in {\mathbb R} \, .
\label{11}
\end{equation}
We shall assume that we are given a relativistic quantum field theory described in terms of a quasi-local algebra ${\cal A}$ satisfying
\begin{itemize}
\item[{\bf A1}] the Haag-Kastler axioms \cite{14};
\item[{\bf A2}] for all $A \in {\cal A}$, $\lim_{t\to0} \Vert \tau_{t}(A)-A \Vert =0$ and $\lim_{\vec{x}\to\vec{0}} \Vert \sigma_{\vec{x}}(A)-A \Vert =0$;
\item[{\bf A3}] for all $A \in {\cal A}$, $\lim_{\vec{v}\to\vec{0}} \Vert \lambda_{\vec{v}}(A)-A \Vert =0$;
\end{itemize}
together with a state $\omega$ defined on ${\cal A}$, which is

\begin{itemize}
\item[{\bf A4}] either a pure state or a factor state;
\item[{\bf A5}] invariant under the automorphism group of space-time translations~$\{ \xi(a) \mid a\in {\mathbb R}^{1+\nu} \}$ and
extremal space translation invariant;
\end{itemize}
In the relativistic case space translations are asymptotically abelian in norm:
\begin{equation}
\label{H28}
\lim_{|\vec{x}|\to\infty} \Vert \lbrack A,\sigma_{\vec{x}}(B) \rbrack \Vert = 0.
\end{equation}
As $\omega $ is, according to {\bf A4-A5}, an extremal space translation invariant factor state, it is
clustering (\cite[Example 4.3.24]{1}):
\begin{equation}
\label{H29}
\omega \bigl(A\sigma_{\vec{x}}(B)\bigr) - \omega (A)\omega (B) \to 0 \quad
\mbox{ as } \quad |\vec{x}|\to\infty.
\end{equation}

Let ${\cal O} \to {\cal A}({\cal O})$ be the net of local algebras in \cite{14}, denote the representation obtained by the GNS construction from the state $\omega$ by $\pi_{\omega}$, and consider the von Neumann rings
\begin{equation}
{\cal R}({\cal O})=\pi_{\omega}({\cal A}({\cal O}))^{''}.
\end{equation}
By \cite[pp.~129--132]{19}, we may choose for the algebra of local observables in ${\cal O} \subset {\mathbb R}^{1+\nu}$ a $C^*$-algebra ${\cal A}_S({\cal O}) \subset {\cal R} ({\cal O})$ such that for all $A \in {\cal A}_S({\cal O})$ the assumption {\bf A2} holds.
Note that {\bf A3} was not a part of Assumption 3.1.2 of \cite{19}, but may be included by an extension of the argument, using the following result of Sakai \cite{20}:  let $\nu$ be the left invariant Haar measure on the orthochronous Poincar\'e group $P^\uparrow_+$  and let
$L^1 (P^\uparrow_+ , \nu)$ be the group algebra of $P^\uparrow_+$. For $f \in L^1 (P^\uparrow_+ , \nu)$ and
$\{ \alpha_g (A)  \mid g \in {\rm supp} \,  f \} \subset {\cal R} ({\cal O}) $,
put
\begin{equation}
\label{sakai-2}
T_f (A) = \int_{P^\uparrow_+} {\rm d} \nu (g) \;  f(g) \alpha_g (A)  ,
\end{equation}
where the integral is defined by using the $\sigma$-weak topology on ${\cal R} ({\cal O})$; then one can easily see that
$T_f (A) \in {\cal R} ({\cal O}) $.
Note that $\alpha_g \bigl(T_f (A) \bigr)$ lies in a larger (but nevertheless strictly) local algebra
${\cal R} (\widehat{\cal O})$ and therefore is well-define. Moreover,
\begin{equation}
\lim_{ P^\uparrow_+ \ni g  \to e} \left\| \alpha_g \bigl(T_f (A) \bigr)  - T_f (A) \right\| =0 \,.
\end{equation}
Thus  $T_f (A) $ as defined in (\ref{sakai-2}) is a smooth element with respect to the Poincar\'e group
automorphisms. Next let us consider $f_1, f_2 \in L^1 (P^\uparrow_+, \nu)$ and
\begin{equation}
\label{sakai-1}
\{ \alpha_g (A_i)  \mid g \in {\rm supp} \,  f_i \} \subset {\cal R} ({\cal O}) , \qquad i=1,2.
\end{equation}
Then (see \cite{20})
\begin{eqnarray}
\left\| \alpha_g \bigl(T_{f_1} (A_1)T_{f_2} (A_2) \bigr)  - T_{f_1} (A_1)T_{f_2} (A_2)  \right\|
\to 0,  
\end{eqnarray}
as $P^\uparrow_+ \ni g \to e$. Thus elements of the form (\ref{sakai-2}) generate a
$*$-subalgebra ${\cal A}_0({\cal O})$  of ${\cal R} ({\cal O})$.
Then by the above consideration,
\begin{equation}
\| \alpha_g (B) - B \| \to 0,  \qquad P^\uparrow_+ \ni g \to e,
\end{equation}
for $B \in {\cal A}_0({\cal O})$. Let ${\cal A}_S({\cal O})$ be the $C^*$-norm closure of~${\cal A}_0({\cal O})$. 
It is easily seen that ${\cal A}_S({\cal O})$ is $\sigma$-weakly dense in ${\cal R}({\cal O})$.
Moreover, it consists of smooth elements: for $C \in {\cal A}({\cal O})$, $ \epsilon >0$,
let $B \in {\cal A}_0({\cal O})$ be an element such that $\| C - B \| \le \epsilon$; then
\begin{eqnarray}
\| \alpha_g (C) - C\| \; \le  \;  & \; \| \alpha_g (C) - \alpha_g (B)\| \nonumber \\ [3mm]
& + \| \alpha_g (B) - B\| \nonumber \\ [3mm]
& + \| B - C\|.
\end{eqnarray}
Hence $\overline{\lim}_{g \to e} \| \alpha_g (C) - C\| \le 2 \epsilon$. Since $\epsilon$ is arbitrary, $\lim_{g \to e}
\| \alpha_g (C) - C \| = 0$.

Therefore the quasi-local algebra defined as
\begin{equation}
{\cal A}=\overline{\bigcup_{{\cal O}\subset {\mathbb R}^{1+\nu}}{\cal A}_S({\cal O})}  \; ,
\label{13}
\end{equation}
where the bar denotes the $C^*$-inductive limit \cite[Proposition 11.4.1]{21}, together with
the automorphisms $\{ \alpha_g \in Aut ({\cal A}) \mid g \in P^\uparrow_+  \}$
forms a $C^*$-dynamical system $({\cal A} , P^\uparrow_+ ,  \alpha_g )$.
Note that the Weyl algebra of the canonical commutation relations (CCR) does not satisfy Assumptions~{\bf A2-A3} --- for the violation of {\bf A2}, see
\cite[Vol.~2, Theorem~5.2.8]{1}.

Important examples included in the above framework are the renormalised vacuum state of the
$P(\phi)_2$ theory \cite{22,23,24,25} and the temperature states of the same theory \cite{26,27}, as we now explain. In order to do that, we need a brief exposition of the barest elements of the theory.

Let ${\mathfrak F} := \bigoplus_{n=0}^{\infty}\bigotimes_{s}^{n}({\cal H})$ be the bosonic Fock space (the subscript $s$ indicates the {\em symmetric} tensor product of copies of ${\cal H}$) over the one-particle space
${\cal H}$ given by $L^{2}({\mathbb R}, {\rm d}k)$; as usual $\bigotimes_{s}^{0}{\cal H}:={\mathbb C}$.
By $\Omega\equiv (1 , 0, \ldots) \in {\mathfrak F}$
we denote the vacuum vector. The free Hamiltonian
\begin{equation}
H_{0}:=d\Gamma(\omega)
\end{equation}
is the second quantisation of the one-particle energy
$\omega(k):= \sqrt{k^{2}+m^{2}}$ with $k\in {\mathbb R}$ and mass~$m>0$, considered as a multiplication operator on ${\cal H}$. The number operator on the Fock space is $N:=d\Gamma( \unit )$. There is a representation of the CCR by creation and annihilation operators
$a^{*}(f)$ and~$a(f)$, $f\in{\cal H}$ (see, e.g., \cite[Section~3.2]{25}). Understanding these objects as operator-valued distributions and writing symbolically
\begin{equation}
a(h) = \int  {\rm d}k \; \overline{h (k)} a(k)  ,
\;  \; a^{*}(h) = \int  {\rm d}k \; h(k)a^{*}(k)  ,
\end{equation}
the free field is given by
\begin{equation}
\phi(x) = \int \frac{{\rm d}k}{\omega(k)^{1/2}} \; {\rm e}^{-ikx} \bigl(a^{*}(k) + a(-k)
\bigr)\, .
\label{14a}
\end{equation}
This expression is again considered as an operator-valued distribution and as such the multiplication of these objects at the same point $x$ is not a well defined operation. We define powers of the fields by ``point splitting'', e.g.,
\begin{equation}
{:}\phi^{2}(x){:} \; \equiv \lim_{y\to x} \; [ \phi(y)\phi(x) - (\Omega,\phi(y)\phi(x)\Omega) ] \; ,
\end{equation}
and similarly for higher powers. Using the CCR, one sees that this leads to the prescription of Wick ordering: all creation operators stand to the left of all annihilation operators.

We fix a real, semi-bounded polynomial of degree $2n$
\begin{equation}
P(\lambda) = \sum_{j=0}^{2n} a_{j}\, \lambda^{j}
\label{14b}
\end{equation}
and choose a function $g\in C_{0}^{\infty}({\mathbb R})$ with $0 \le g(x) \le1 $. Define the interaction Hamiltonian localised in a compact space region by
\begin{equation}
V(g) = \int {\rm d}x \; g(x) \,  {:}P(\phi(x)){:} \, .
\label{14c}
\end{equation}
The Wick ordering in the powers of field operators makes~$V(g)$ a well-defined unbounded quadratic form. By investigating the smoothness and symmetry properties of the scalar kernel of $V(g)$, it is seen that $V(g)$ is an unbounded, symmetric operator with domain contained in $D(N^{n})$. By \cite[Theorem 6.4]{25}
\begin{equation}
H(g) = H_{0} + V(g)
\label{14d}
\end{equation}
is essentially self-adjoint on $D(H_{0})\cap D(V(g))$. Moreover, $H(g)$ is semibounded from below. Let
\begin{equation}
E_{g} := \mbox{ inf } \lbrace\mbox{Spectrum } H(g)\rbrace \, .
\end{equation}
The corresponding eigenvalue is an isolated eigenvalue of~$H(g)$ with multiplicity one \cite{23}, thus 
corresponding to an eigenvector $\Omega_{g}\in {\mathfrak F}$, $\Vert \Omega_{g} \Vert = 1$, such that
$H(g)\Omega_{g} = E_{g} \Omega_{g}$. As~$g(\cdot)\to 1$ the scalar product $(\Psi, \Omega_{g}) \to 0$ for all
$ \Psi \in {\mathfrak F}$. This clearly demonstrates that Hilbert space methods are insufficient,
and thus justifies the operator algebraic framework, which allows us to
obtain (see \cite{23}) the vacuum state $\omega$ as the $w^*$-limit, as $g(\cdot)\to 1$, of the states
\begin{equation}
\omega_{g}(A) = (\Omega_g, A \Omega_g) \, ,
\label{15}
\end{equation}
with $A$ in the $C^*$-closure ${\cal A}$ of the local von Neumann algebras ${\cal R}({\cal O})$ generated
by the Weyl operators
\begin{equation}
W(f)=\exp \left( i \int {\rm d}x \, f(x) \phi(x) \right),
\end{equation}
with
\begin{equation}
f \in \overline
{\omega^{1/2} {\cal D}_{\mathbb R}({\cal O}) + i \omega^{-1/2}
{\cal D}_{\mathbb R}({\cal O})}^{L^2 ( \mathbb R, {\rm d}x)}.
\end{equation}
For the thermal field theory, the free Liouvillean $L_{\omega}$ (see~(\ref{3})) is the Araki-Woods
Liouvillean~$L_{\scriptscriptstyle AW}$ (see \cite{12} and \cite{26}).  Euclidean techniques can be used to define
the operator sum
\begin{equation}
H_{\beta} (g) := L_{\scriptscriptstyle AW} + \int {\rm d}x \; g(x) : P(\phi_{\beta}(x)):_{C_{\beta}} \,  ,
\label{16a}
\end{equation}
where the Wick ordering is defined from (\ref{14b}) in terms of the thermal covariance 
function~$C_{\beta}$~\cite{12,26}, as an essentially self-adjoint operator \cite{27}: let its closure be defined by the same symbol.
The vector
\begin{equation}
\Omega_{\beta} (g) := \frac{ {\rm e}^{- \frac{\beta}{2} H_{\beta}(g)} \Omega_{\scriptscriptstyle AW}}
{\Vert {\rm e}^{-\frac{\beta}{2} H_{\beta}(g)} \Omega_{\scriptscriptstyle AW} \Vert} ,
\label{16b}
\end{equation}
where $\Omega_{\scriptscriptstyle AW}$ is the cyclic GNS vector associated to the Araki-Woods state, induces a KMS state $\omega_{\beta}(g)$ for the $W^{*}$-dynamical system $(\pi_{\scriptscriptstyle AW}({\cal A})'', \tau^{g})$, where $\pi_{\scriptscriptstyle AW}$ is the Araki-Woods representation and $\tau_{t}^{g}$ is the time evolution
\begin{equation}
\tau_{t}^{g}(A):=  {\rm e}^{itH_{g}} A {\rm e}^{-itH_{g}} \, , \qquad A \in {\cal A} \, .
\label{16c}
\end{equation}
Note that $H_{g}$ and $ H_{\beta}(g)$ induce the same group of automorphisms
on ${\cal A}$, thus there is no $\beta$ dependence on the level of automorphisms.
It was proved in \cite{18} that the limit
\begin{equation}
\omega_{\beta} := \lim_{g(.) \to1} \omega_{\beta}(g)
\label{16d}
\end{equation}
exists and defines a state on
\begin{equation}
{\cal A} := \overline{\bigcup_{{\cal O} \subset {\mathbb R}^{1+\nu} }{\cal R} ({\cal O})} .
\label{16e}
\end{equation}
The $C^*$-algebra (\ref{16e}) is isomorphic to the $C^*$-inductive limit of
the local von Neumann algebra~${\cal R}_{\scriptscriptstyle AW}({\cal O})$ generated by the Weyl operators 
in the Araki-Woods representation~\cite{12,26,27}. The thermal states $\omega_{\beta}$, $ \beta >0$,  
defined by (\ref{16d}) satisfy a relativistic generalisation of the KMS condition (see \cite{12} and references given there).

We now turn to the question of whether assumptions {\bf A1-A5} are satisfied for the above-mentioned examples.
For~ {\bf A4}, {\bf A5} (except purity), see \cite{22} and~\cite{23}. For the unicity of the vacuum (the purity in {\bf A4}) see \cite{24} and references given there. Property {\bf A5}   for the thermal state were proved in \cite{18}. The replacement of ${\cal R}({\cal O})$ by a weakly dense subalgebra such as the one of the form~(\ref{sakai-2}) does not alter the validity of the above-mentioned results, see the remarks after the definition pg.~399 of~\cite{22}. Thus {\bf A5} may also be assumed to hold for these examples, in particular the replacement of (\ref{16e}) by~(\ref{13}).
  Finally, for the factoriality property in {\bf A4}, we may decompose the thermal state into factor states through the primary decomposition~\cite[Vol.~2, Theorem~5.3.30, pg.~116]{1}, and pick any one of the latter as our state. For the thermal~$P(\phi)_2$ theory it is expected that the KMS state is unique and thus a factor state, but a proof is still missing. It would be interesting to know whether there exist non-equilibrium stationary states (NESS) which satisfy the properties {\bf A1-A5}, but so far we are not aware of any rigorously constructed NESS for interacting relativistic quantum field theories.

\section{The main theorem}
\label{Sec3}

We want to inherit clustering properties of the time translation from those of space translation with the help of a smearing effect. First we specify the properties of the smearing functions. Let $f, g$ be  $C^{\infty}$-functions of compact support, such that
\begin{equation}
\int_{-\infty}^{\infty} {\rm d} x \; f(x) = 1, \qquad f(x) \ge 0, 
\end{equation}
and
\begin{equation}
{\rm supp} \, f \in
[ -( a-\delta) , a-\delta ],
\label{18a}
\end{equation}
where $\delta<a$. Set
$\alpha_{t} := t^{-1/2-\epsilon},
$
with
$0 < \epsilon < 1/2$,
and define
\begin{equation}
g_{t}(v) :=  \frac{1}{\alpha_{t}}g \left(\frac{v}{\alpha_{t}} \right)
\qquad
\forall v\in {\mathbb R}\, .
\label{18d}
\end{equation}
Clearly $g_{t}(v) \to \delta (v)$ as $t \to \infty$, thus $\{ g_t \}_{t >0}$ is an
approximation of the Dirac delta function. For later usage we define also
an approximation  $\{\hat f_s \}_{s>0}$ of the Dirac delta function,
whose convergence is less rapid:
\begin{equation}
\hat f_{s}(v) :=  \frac{1}{ \hat \alpha_{s}}g \left(\frac{v}{\hat \alpha_{s}} \right), 
\qquad
v\in {\mathbb R}\, ,
\label{18d2}
\end{equation}
with $\hat \alpha_{s} := s^{-1/4-\epsilon/4}$,
$0 < \epsilon < 1/4$.

\bigskip

Our main theorem may now be stated:

\begin{theorem}
\label{Th1}
Let a relativistic quantum field theory satisfy the Assumptions {\bf A1-A3} of
Section~\ref{Sec2} and let $\omega $ be a state satisfying the Assumptions {\bf A4-A5}.
Then
\begin{itemize}
\item [(i)] $ \omega _v  \equiv \omega \circ \lambda _v $ is  an extremal space translation invariant and time invariant factor state;

\item[(ii)] $\omega _f :=\int {\rm d} v \, f(v)\omega _v$ is not a factor state. Its centre contains space- and time-translation invariant elements $B_f$ satisfying the following three properties:
\begin{itemize}
\item [(ii.a)] $ \omega _f(AB_f)=\omega _f(B_f A) $;
\item [(ii.b)] $ \lim _{x\rightarrow \infty }\omega _f \bigl(A  \sigma _x (B-B_f) \bigr)=0$;
\item [(ii.c)]  $ \lim _{t\rightarrow \infty }\omega _f \bigl(A  \tau _t (B-B_f) \bigr)=0 $.
\end{itemize}
\end{itemize}
\end{theorem}

\begin{proof}
\begin{itemize}
\item[(i)] by assumption $\omega $ is invariant under space-time  translations:
\begin{equation}
\omega \circ\xi_{ (t, \vec x) } = \omega \qquad \forall (t, \vec x ) \in {\mathbb R}^{1+\nu}.
\end{equation}
From (\ref{11})
we conclude that
\begin{equation}
\label{H23modified}
\lambda_{v} \circ \xi_{a} =\xi_{L(v)a } \circ \lambda_{v}.
\end{equation}
Now consider the state $\omega _v = \omega \circ \lambda _v$. Clearly
\begin{equation}
\begin{array}{rl}
\omega _v \circ \xi_{(t,\vec x)} & =
\omega \circ \lambda _v \circ \xi_{(t,\vec x)}
\\ [4mm]
&
= \omega  \circ \xi_{L(v)(t,\vec x)} \circ \lambda _v
= \omega   \circ \lambda _v
\\ [4mm]
&
= \omega_v \qquad \forall (t, \vec x ) \in {\mathbb R}^{1+\nu}.
\end{array}
\end{equation}
This shows that $\omega_v$ is also space translation invariant and time invariant.
Now assume that $\omega _v$ allows a decomposition into space-translation invariant states.
Then we could use the
group property to derive a  decomposition of $\omega$ into space-translation invariant states,
which is not allowed. Thus $\omega _v$ is extremal space translation invariant.
As $\lambda _v$ is an automorphism of ${\cal A}$, the state $\omega _v$ is a factor state, just like $\omega$.

\item[(ii)] Smearing out the state $\omega$ (i.e., forming a convex combination)
will lead to a non-trivial centre, thus we expect that $\omega _f$ fails to be a factor state. Indeed, given some $ B \in {\cal A}$
we can weakly define a non-trivial element
$B_f$, which lies in the centre:
\begin{equation}
\label{neu38}\omega_f (AB_f C) :=  \int {\rm d}v \; f(v)\, \omega _v (AC) \omega _v (B)
\end{equation}
for all $A, B, C \in {\cal A}$.
Clearly (\ref{neu38}) implies $\omega_f (AB_f C) = \omega_f (AC B_f ) $ and $\omega_f (B_f C) = \omega_f (C B_f ) $.
Next we show that (ii.b) holds:
\begin{eqnarray}
\label{new43}
& \lim _{x\rightarrow \infty }\omega _f \bigl(A  \sigma _x (B-B_f) \bigr) = \qquad \qquad 
\\ [1mm]
& \qquad \qquad =
\lim _{x\rightarrow \infty }\int {\rm d}v  f(v)  \omega _v \bigl(A\bigr(\sigma _x (B)-\omega _v(B) \unit \bigr)\bigr)  .
\nonumber 
\end{eqnarray}
We have made use of the fact that $\omega _v(\sigma _x (B)) = \omega _v(B)$ for all $x \in {\mathbb R}^s$.
The r.h.s.~in~(\ref{new43}) vanishes, as  $w$-$\lim _{x\rightarrow \infty }  \sigma _x (B) = \omega _v(B)  \unit$.

Next we prove (ii.c). We proceed in several steps.

\begin{itemize}
\item [(a)] Different from \cite{8} the initial state will, except for the ground state or for the tracial state (provided it exists) not be invariant under Lorentz boosts. In order to use~(\ref{9b}) and shift it between the operators $A$ and $B$ we have to control that the effect of the state is negligible. We define
\begin{equation}
\label{H19a}
h_t(v):=\int {\rm d}v_1 \; f(v-v_1) \, g_t (v_1),
\end{equation}
with $g_{t}(v) $ as described in (\ref{18d}).
It follows from the fact that the width of the ${\rm supp \ } g_t$ is proportional to $t^{-1/2-\epsilon}$
that
\begin{equation}
\label{H20}
\bigl| h_t(v)-f(v) \bigr| \leq \, c \, \sup \left| \; \frac{{\rm d} }{{\rm d}v} \, f \; \right|
\cdot t^{-1/2-\epsilon} .
\end{equation}
This allows to approximate (ii.c) by
\begin{eqnarray}
\label{H21}
&\int {\rm d}v \; f(v)\int dv_1 \; g_t (v_1)  \times \qquad \qquad 
\\ [1mm]
&\qquad
\qquad 
\qquad 
\times 
\omega _v
\Bigl( \lambda_{v_1} (A) \bigl(\lambda_{v_1} \circ \tau _t(B)-\omega _v(B) \unit \bigr) \Bigr).
\nonumber 
\end{eqnarray}
Note that we have shifted the integration over $v$ by $v_1$.

\item [(b)] Let us now concentrate for a moment on the term
\begin{equation}
\label{H21a}
\lambda_{v_1} (A) \bigl(\lambda_{v_1} \circ \tau _t(B)-\omega _v(B) \unit \bigr) .
\end{equation}
From (\ref{9b}) we conclude that
\begin{equation}
\qquad
\qquad
 \lambda_{v_1} \circ \tau_{t} (B)= \tau_{ t \cosh v_1} \circ \sigma_{ t \sinh v_1} \circ \lambda_{v_1} (B).
\label{H23}
\end{equation}
Thus (\ref{H21a}) equals
$\lambda_{v_1} (A) \bigl(\tau_{ t \cosh v_1} \circ \sigma_{ t \sinh v_1}
\circ \lambda_{v_1} (B)
-\omega _v(B) \unit \bigr) $.
Since $g_t$ has shrinking compact support for $t \to \infty$, there
exist constants  $c_A$ and $c_B$, which may depend on~$A$ and $B$, respectively, such that
\begin{equation}
\label{H22}
\| \lambda_{v_1}(A)-A \|   \leq
c_A \, t^{-\epsilon -1/2},
\end{equation}
\begin{equation}
\| \lambda_{v_1}(B)-B \|   \leq
c_B \, t^{-\epsilon -1/2} ,
\end{equation}
for all $v_1 \in {\rm supp}  \, g_t$.
Thus (\ref{H21a})  can be approximated in norm:
\begin{widetext}
\begin{equation}
\label{H23-2}
\begin{array}{rl}
& \Bigl\| \lambda_{v_1} (A) \bigl(\tau_{ t \cosh v_1} \circ \sigma_{ t \sinh v_1}
\circ \lambda_{v_1} (B)
-\omega _v(B) \unit \bigr)
\\ [4mm]
&\qquad \qquad \qquad \qquad
-
A  \bigl(\tau_{ t \cosh v_1} \circ \sigma_{ t \sinh v_1}
(B)
-\omega _v(B) \unit \bigr) \Bigr\|
\\ [4mm]
& \qquad \le \Bigl\| (\lambda_{v_1} (A) - A)   \bigl(\tau_{ t \cosh v_1} \circ \sigma_{ t \sinh v_1}
\circ \lambda_{v_1} (B)
-\omega _v(B) \unit \bigr) \Bigr\|
\\ [4mm]
&\qquad \qquad + \Bigr\| A    \Bigl(\tau_{ t \cosh v_1} \circ \sigma_{ t \sinh v_1}
\circ \lambda_{v_1} (B) - \tau_{ t \cosh v_1} \circ \sigma_{ t \sinh v_1}
(B) \Bigr) \Bigr\|
\\ [4mm]
&\qquad  \le c_A \, t^{-\epsilon -1/2} \cdot 2 \| B \| + \| A \| \cdot c_B \, t^{-\epsilon -1/2} .
\end{array}
\end{equation}
\item [(c)] Using the time invariance of $\omega_v$ we find that
\begin{equation}
\label{H23a}
\begin{array}{rl}
& \left| \omega _v
\bigl( \lambda_{v_1} (A) \bigl(\lambda_{v_1} \circ \tau _t(B)-\omega _v(B) \unit \bigr) \bigr) -
\omega _v
\bigl( \tau_{ -t \cosh v_1} (A)( \sigma_{ t \sinh v_1} (B) -\omega _v(B) \unit )  \bigr)
\right|
\\ [4mm]
&\qquad  \le c_A \, t^{-\epsilon -1/2} \cdot 2 \| B \| + \| A \| \cdot c_B \, t^{-\epsilon -1/2} .
\end{array}
\end{equation}

\item [(d)] So far the considerations have been the same for Galilei invariant time evolutions as for relativistic time evolutions. In order that the $v_1$ integration is only connected with the space translation we need an additional estimate:
\begin{equation}
\label{H25}
\sup_{v_1 \in {\rm supp} \, g_t} \|  \tau_{t}(A) - \tau_{ t \cosh v_1}(A) \|
\leq c'_A\, | t| ^{-2\epsilon},
\end{equation}
with a constant $c'_A$ which may depend on~$A$.
This follows from expanding $\cosh v_1$ in a power series and taking the
support properties of $g_t$ into account.
Thus for~$t >0$
\begin{equation}
\label{H25a}
\begin{array}{rl}
& \left| \omega _v
\bigl( \lambda_{v_1} (A) \bigl(\lambda_{v_1} \circ \tau _t(B)-\omega _v(B) \unit \bigr) \bigr) -
\omega _v
\bigl( \tau_{ -t } (A)( \sigma_{ t \sinh v_1} (B) -\omega _v(B) \unit )  \bigr)
\right|
\\ [4mm]
&\qquad  \le c_A \, t^{-\epsilon -1/2} \cdot 2 \| B \| + \| A \| \cdot c_B \, t^{-\epsilon -1/2} + c'_A\, t^{-2\epsilon}
\cdot 2 \| B \| .
\end{array}
\end{equation}

We can now replace (ii.c) by estimating
\begin{eqnarray}
\label{H26}
& \left|\int {\rm d}v_1 \; g_t (v_1) \, \omega _v \left(\tau_{-t}(A) \bigl(\sigma_{t\sinh v_1}(B)-\omega _v(B) \unit \bigr) \right) \right|
\nonumber \\ [4mm]
& \qquad  \leq
\| A \|  \; \Bigl\{ \int {\rm d}v_1 \, {\rm d}v_2 \;
g_t (v_1)
g_t (v_2) 
\omega _v
\Bigl( \bigl( \sigma_{t \sinh v_1} (B^*) - \omega _v (B) \unit  \bigr)
\bigl(\sigma_{ t \sinh v_2} (B) -\omega _v(B^*) \unit \bigr) \Bigr) \Bigr\}^{1/2}.
\end{eqnarray}
\end{widetext}
\item [(e)] Returning to (ii.c) we have still to integrate over~$v$. As $\omega _v$ is an extremal space translation invariant primary state, $\omega _v$ is clustering (\cite{1}, Example 4.3.24):
\begin{equation}
\label{H29b}
\omega_v \bigl(A\sigma_{\vec{x}}(B)\bigr) - \omega_v (A)\omega_v (B) \to 0  
\end{equation}
as $|\vec{x}|\to\infty$. With  (\ref{18a}), (\ref{18d}) the integral $ \int {\rm d} v_{1}...$ in (\ref{H21}) has an upper bound
$||B||(||A||+1)$ which is independent of $t$ and~$v$. For all  $v$ the $v_{1}$ integral converges to zero as
$t\to \infty$. This follows from (\ref{H26}) and (\ref{H29b}) and Lebesgue dominated convergence: one makes in the double integral on the r.h.s.~of~(\ref{H26}) the change to variables $w_{1}=v_{1}/\alpha_{t}$ and $w_{2}=v_{2}/\alpha_{t}$ suggested by (39). Using space translation invariance of $\omega_{v}$ and the choice of $\alpha_{t}$, the integrand in the new variables is seen to tend to zero by (\ref{H29b}) for all $(w_{1},w_{2})$ except along the diagonal $w_{1}=w_{2}$ (thus a set of zero Lebesgue measure in $\mathbf{R}^2$), and is, uniformly in $t$, bounded by $g(w_{1})g(w_{2})(2||B||)^2$.
Since~$f$ has compact support and (\ref{18a}) is satisfied,  a second
application of the Lebesgue dominated convergence theorem finally proves the result.
\end{itemize}
Collecting all estimates proves the theorem. \qed
\end{itemize}
\end{proof}

\bigskip
We now want to draw conclusions for the state $\omega $ we started with:

\begin{theorem}
\label{Th2}
In QFT a primary state that is extremal space invariant and also time invariant is also extremal time invariant.
\end{theorem}

\begin{proof}
We start with such a state and smear it to obtain~$\omega _f$. By Theorem \ref{Th1}
(ii.c) a time translated operator converges
weakly to an element in the centre. The decomposition into extremal time-invariant states corresponds to a maximal abelian subalgebra within the commutant, which contains all time invariant elements. This algebra is contained in the centre, and coincides with the algebra of space-translation invariant elements. This follows from Theorem \ref{Th1}
(ii.b). Therefore also the decomposition into extremal time invariant states coincides with the decomposition into extremal space translation invariant states \cite{28}.
\qed
\end{proof}

Unfortunately Theorem \ref{Th2} is not strong enough to guarantee the mixing property~(\ref{6}). There is still the possibility that the time invariant operators in the commutante are not weak limits but only invariant means. Notice that in our proof we used however just limits.

\bigskip

Another  possibility to interpret Theorem \ref{Th1} is by varying $f(v)$, in a sense to be made precise below.

\begin{theorem}
\label{Th3}
Let $\omega$ be a state satisfying A1-A5. Then for any state $\omega_{f}$ of the form (ii) of Theorem 3.1 (hence space and time tranlation invariant), among which there are some arbitrarily close to $\omega$ in the weak* topology,  the time evolution is weakly asymptotically abelian:
\begin{equation}
\lim_{t\to\infty} \omega_{f}\bigl( A\lbrack\tau_{t} (B) ,C\rbrack D \bigr) = 0
\label{H30}
\end{equation}
for all $A,B,C,D \in \cal{A}$.
\end{theorem}

\begin{proof}
Let  $\epsilon>0$ and $A\in\cal{A}_S$ be given. We may then choose $\delta>0$,  depending on $A$ and~$\epsilon$,   such that
for $|v|<\delta$, $\|\lambda_{v}(A)-A\|< \epsilon$. Choose now $f$ smooth with compact support in $ [-\delta,\delta]$
such that
$\int_{-\infty}^{+\infty} {\rm d} v f(v)=1$.
Then
$$
\left| \int {\rm d} v \; f(v)\omega_{v}(A)-\omega(A) \right|<\epsilon \, ,
$$
and thus a proper choice of $f$ makes $\omega_{f}$ arbitrarily close to $\omega$ in the weak* topology. In the representation $\pi_{\omega_{f}}$ corresponding to any  state of the form (ii) of Theorem~\ref{Th1} (close to $\omega$ or not), (iic) of Theorem~\ref{Th1} asserts that~$\tau_{t}(B)$ converges to an element of the centre, for any~$B\in\cal{A}_S$. Thus (\ref{H30}) holds.
\qed
\end{proof}

If, however, we try to obtain a statement that is close to (\ref{6}) for an individual state, we have the possibility to scale the smearing function $f$ and consider the limit of the states
\begin{equation}
\omega _{\hat f_s} :=\int {\rm d} v \, \hat f_s (v)\omega _v
\end{equation}
as $s \to \infty$. Note that $\hat f_s$ was defined in (\ref{18d2}).
Clearly
\begin{equation}
\lim _{s\rightarrow \infty }\omega_{\hat f_s} (A)=\omega (A), \qquad A \in {\cal A} .
\end{equation}
We now would like to estimate
\begin{equation}
\Bigl| \omega_{\hat f_s}(A\tau_t(B))-\omega_{\hat f_s} (A)\omega_{\hat f_s} (B) \Bigr|
\end{equation}
for $s$ and $t$ large.
Setting $s=t$ (i.e., taking the limits $t \to \infty$ and $s \to \infty$ simultaniously)
we find that (\ref{H20}) changes to
\begin{eqnarray}
\label{62}
\bigl| \hat h_t (v) - \hat f _t(v) \bigr|
& \le
\, c \, \sup_v \left| \frac{{\rm d}}{{\rm d}v} \hat f_t (v) \right| \cdot t^{- \epsilon-1/2}
\nonumber
\\
& \le c \,  \sup_v \left|  f'(v) \right| \cdot t^{- \epsilon/4}, 
\end{eqnarray}
for $0 < \epsilon < 1/4$.
The remaining estimates remain unchanged.
Thus
\begin{equation}
\label{H31b}
\lim _{t\rightarrow \infty }\omega _{\hat f_t}
\bigl(A \tau _t(B)\bigr) =\omega (A)\omega (B) ,
\end{equation}
for all $f \in C^\infty ({\mathbb R})$ which satisfy (\ref{18a}), i.e.,
$\int_{-\infty}^{\infty} {\rm d} x \; f(x) = 1$, $ f(x) \ge 0$ and ${\rm supp} \, f \in
[ -( a-\delta) , a-\delta ]$, where $\delta<a$. One can see from (\ref{62}) that
we are unable to control the limit $a \to 0$ with our estimates.

We note that we could still play with the support of the initial function $f$.
As a last remark, if $\omega$ is Lorentz-invariant, we replace $\omega_{v}$ by $\omega$ in (45), which, together with (e), yields the mixing property (6): this is Maison's result \cite{7}.

\begin{remark} The relation between boost and shift was used by D.~Buchholz~\cite{29} to show that for a large class  of states $\{ \hat \omega \}$ (which are normal w.r.t.~the vacuum state) the weak$^*$ limit points  of the nets $\{ \hat \omega \circ \tau_t \}_{t >0}$ as $t \to \infty$ are states, which are
invariant under spatial translations (and thus vacuum states). It should be possible to extend this result to a class of states of the form
\begin{equation}
\hat \omega (.) = \omega (A \; . \; D)
\end{equation}
with $\omega$ satisfying {\bf A4-A5} and $A, D \in {\cal A}$ operators with compact energy support (as defined, e.g., in \cite{30}).
Thus if $\omega$ is the only space-translation invariant state
w.r.t.~the representation~$\pi_\omega$, then the  weak$^*$ limit of all the states in the class $\{\hat \omega\}$ would be
$\omega$ itself.
Whether such a result actually holds and is related to our work has to be further investigated.
\end{remark}

\section{Conclusion, open problems and conjectures}
\label{Sec4}

We have shown that, given  a distinguished state of a quantum field
$\omega$, which satisfies the assumptions {\bf A4-A5} of Section~\ref{Sec2},  there exist  space-time translation invariant states, some of which are arbitrarily close to $\omega$ in the weak* topology, for which the time evolution is weakly asymptotically abelian.

The proof depends on two features: the fact that the Lorentz boost relates space translations and time translations,  and, secondly, locality~(\ref{H28}), which implies the existence of good space-like cluster properties (\ref{H29}). These properties are valid for both ground states and primary (factor) thermal states.

Assumptions {\bf A1-A5} hold for the ground state and thermal state(s) of the $P(\phi)_2$ model, but they are expected to hold for any relativistic quantum field theory, and in this sense we have strengthened the conjecture in~\cite{8} and \cite{31} that the observables 
of a Poincar\'e-invariant theory are a mixing system.

It follows from our results that these states, provided the corresponding GNS vector is separating (such states are called {\em modular states} in the literature),  possess the properties of return to equilibrium
%
%
\begin{equation}
\lim _{t\rightarrow \infty }\omega \bigl( A^*\tau _t(B)A \bigr)=\omega (A^*A)\omega (B)
\end{equation}
and corresponding consequences  which are usually rather hard to prove. This fact may be regarded as a bonus from quantum theory,
%
%
but, more specifically, of quantum field theory. Indeed, for quantum spin systems, the property of weak asymptotic abelianness 
(Theorem~\ref{Th3}),
%
%
is not generally valid \cite{32,33}.

The above reference to quantum field theory  includes Galilean-invariant
theories~\cite{8,10}. One basic difference between fully relativistic quantum field theories and the latter is that only the former display vacuum polarization, which leads to non-Fock representations of the CCR because of Haag's theorem \cite[pg.~55]{19}. There is, however, an advantage of Galilean invariant theories over Poincar\'e-invariant ones: at least for the class considered in~\cite{8,10}, a stronger state-independent mixing property, namely
\begin{equation}
\lim_{t\to\infty} \Vert A\tau_{t} (B) \Vert = \Vert A \Vert \Vert B \Vert
\qquad
\forall A,B \in{\cal A},
\label{40}
\end{equation}
may be proven \cite{8, 9}. Quantum mixing systems in the sense of (\ref{40}) may be shown to be indeterministic and undecidable in a quite precise sense \cite[Lemma L4]{31}, with an interesting application: unpredictability of the symmetry breaking in a phase transition such as the one which presumably occurred after the big-bang \cite{31}.
It is an open problem to prove (\ref{40}) for our class of systems.

On the other hand, relativistic quantum fields and their equilibrium states play an important role in applications in cosmology, in particular in the dark energy
problem~\cite{34}. In cosmology, thermal quantum fields associated to the temperature of background radiation (presently of about 3K) in the hot big-bang model (see, e.g., \cite[pg.~187]{35}) must be of special relevance, which is, yet, to be fully explored (see, however, \cite{36}).

\bigskip

{\em Acknowledgements.}
Two of us (C.J. and W.F.W.) are grateful to Prof.~Jakob Yngvason for an invitation to stay at the Erwin
Schr\"odinger Institut (ESI) during the period June~1--14, 2009, which made this collaboration possible. We all thank Geoffrey Sewell for illuminating discussions.

\end{document}